\pdfoutput=1
\newif\ifFull
\Fulltrue
\ifFull
\documentclass[11pt]{article}
\usepackage{amsthm}
\usepackage{cite}
\usepackage{fullpage}
\else
\documentclass[sigconf, natbib]{acmart}
\usepackage{balance}
\fi
\usepackage{graphicx,subcaption}
\usepackage[noend]{algorithmic}
\usepackage{url,color}
\usepackage{amsmath}

\newcommand{\Get}{{\textsf{get}}}
\newcommand{\Put}{{\textsf{put}}}
\newcommand{\Read}{{\textsf{read}}}
\newcommand{\Write}{{\textsf{write}}}

\newtheorem{theorem}{Theorem}
\ifFull
\newtheorem{lemma}[theorem]{Lemma}
\fi

\begin{document}


\ifFull
\title{BIOS ORAM: Improved Privacy-Preserving Data Access \\ 
for Parameterized Outsourced Storage}
\fi

\ifFull
\author{Michael T.~Goodrich \\
{University of California, Irvine} \\
{Dept. of Computer Science} \\
\texttt{goodrich@acm.org}
}
\fi

\date{}

\maketitle


\ifFull
\begin{abstract}
Algorithms for
oblivious random access machine (ORAM) simulation 
allow a client, Alice, to obfuscate a pattern of 
data accesses with a server, Bob, who is maintaining Alice's
outsourced data while trying to learn information about her data.
We present a novel ORAM scheme that 
improves the asymptotic I/O overhead of 
previous schemes for a wide range of size parameters for
client-side private memory and message blocks,
from logarithmic to polynomial.
Our method achieves statistical security for hiding Alice's access
pattern and, with high probability, achieves an I/O overhead that 
ranges from $O(1)$ to $O(\log^2 n/(\log\log n)^2)$, depending on these size
parameters, where $n$ is the size of Alice's outsourced memory.
Our scheme, which we call BIOS ORAM, combines multiple uses of B-trees with
a reduction of ORAM simulation to isogrammic access sequences.
\end{abstract}
\fi

\section{Introduction}
In \emph{outsourced storage} applications, a client,
Alice, outsources her data to a server, Bob, 
who stores her data and provides her with an interface to access
it via a network.
We assume that Bob is 
``honest-but-curious,'' meaning that Bob is trusted to keep Alice's
data safe and available, but he wants to learn as much
as he can about Alice's data.
For privacy protection, we also
assume that Alice encrypts her data by a semantically secure encryption
scheme so that each time she accesses an item then she securely re-encrypts it
before returning it to Bob's storage.
The remaining problem, then, is for Alice to obscure her data access pattern
so that Bob can learn nothing from her access sequence.

Fortunately,
there is a large and growing literature on algorithms for
oblivious random access machine (ORAM) simulation to obfuscate Alice's 
access sequence
(e.g., 
see~\cite{Goldreich:1996,dmn-11,Goodrich2011,Kushilevitz:2012,Stefanov:2013,boneh2011remote,Ohrimenko2014,Apon2014,Stefanov:2013:MOS:2508859.2516673,Shi2011,Goodrich:2012:POS,Goodrich:2011:ORS:2046660.2046680}).
Such ORAM simulation methods provide ways for Alice to 
replace each of the data accesses
in her algorithm, $\mathcal A$, 
with a sequence 
of accesses to her data stored with Bob so as to obfuscate
her original access sequence.
Ideally, such an ORAM scheme achieves \emph{statistically security},
which intuitively means that Bob is unable to determine any information about
Alice's original access sequence beyond its length, $N$,
and the size of her data set, $n$.
In addition to the parameters, $n$ and $N$, 
the following parameters are also important in this context:
\begin{itemize}
\item
$B$: The maximum number of words in a 
message block sent from/to Alice in one I/O operation. 
\item
$M$: The number of words in Alice's client-side private memory.
\end{itemize}

In this paper, we are interested in scenarios where 
$B$ and $M$ can be set to arbitrary values that 
are at least logarithmic in $n$, and at most a constant
fraction of $n$, so as to apply to a wide range of 
scenarios, while still applying to cases
where Alice would still be motivated to outsource her memory to Bob.
Formally, in an ORAM framework,
we assume that Alice's RAM algorithm, $\mathcal A$, 
indexes data using integer addresses
in the range
$[0,n-1]$, using the following operations:
\begin{itemize}
\item 
$\Write(i,v)$: Write the value $v$ into the memory cell indexed
by the integer, $i$.
Since $\mathcal A$ is a RAM algorithm, we assume $i$ and $v$ each fit in 
a single memory word.
\item
$\Read(i)$: Read and return the value, $v$, stored in the cell with
integer address, $i$. 
\end{itemize}
These operations are the low-level accesses that are issued during
Alice's execution of her RAM algorithm, $\mathcal A$.
The goal of
an ORAM simulation scheme is to allow Alice to perform her algorithm,
$\mathcal A$, but to replace each individual read or write operation
with a sequence of input/output (I/O)
messages exchanged between Alice and Bob so that Alice effectively
hides $\mathcal A$'s pattern of
data accesses, i.e., $\mathcal A$'s sequence of read/write operations.
Moreover, we would like to minimize the amortized number of such messages
exchanged between Alice and Bob while still preserving Alice's privacy.

A related concept is that of an \emph{oblivious storage} (OS),
e.g., see~\cite{Apon2014,boneh2011remote,Goodrich:2012:POS,Ohrimenko2014,%
Stefanov:2013:MOS:2508859.2516673,s6547114}.
In this framework, Alice stores a dictionary 
at the server, Bob, of size at most $n$, and her algorithm, $\mathcal A$,
accesses this dictionary using the following operations:
\begin{itemize}
\item 
$\Put(k,v)$: Add the key-value item, $(k,v)$.
An error occurs if there is already an item with
this key.
We assume here that each key, $k$, fits in a single memory word
and each value, $v$, fits in a message block of size $B$.
\item
$\Get(k)$: Return and remove the value, $v$, associated with the key, $k$. 
If there is no item in the collection with key $k$, then return 
a special ``not found'' value. 
\end{itemize}


Note that OS includes ORAM as a special case. For example, we can initialize 
$\mathcal A$'s memory, $A[0..n-1]$, as $\Put(i,A[i])$, for $i=0,\ldots,n-1$. 
Then we can perform any $\Write(i,v)$ as $\Get(i),\Put(i,v)$,
and we can perform any $\Read(i)$ as $v=\Get(i),\Put(i,v)$.
We assume that Alice's original 
access sequence has a given length, $N$, 
where $N$ is at most polynomial in $n$, and an OS scheme replaces
each such operation with a sequence of operations that obfuscate the original
access sequence.
Ideally, Bob should learn nothing about Alice's original access sequence,
which we formalize in terms of a security game.
Let $\sigma$ denote a sequence of $N$
$\Read$/$\Write$ operations (or $\Get$/$\Put$ operations).
An ORAM (resp., OS) scheme
transforms $\sigma$ into 
into a sequence, $\sigma'$, of access operations.
As mentioned above, we assume that each item is stored
using a semantically-secure encryption scheme,
so that independent of whether Alice wants to keep a key-value item
unchanged, the sequence $\sigma'$ involves 
always replacing anything Alice accesses
with a new encrypted value so that Bob is unable to tell if the underlying 
plaintext value has changed.
The security for an ORAM (or OS) simulation is defined
in terms of the following security game.
Let $\sigma_1$ and $\sigma_2$ be two different RAM-algorithm
or dictionary access sequences, of length $N$, 
for a key/index set of size $n$,
that are chosen by Bob and given to Alice.
Alice chooses uniformly at random one of these sequences and transforms
it into the access sequence $\sigma'$ according to her ORAM (or OS) scheme,
which she then executes.
Her ORAM (resp., OS) scheme is \emph{statistically secure} if 
Bob can determine which sequence, $\sigma_1$ or $\sigma_2$, Alice chose
with probability at most $1/2$.
This assumes that Bob learns nothing from the encryption of Alice's
data.

The \emph{I/O overhead}
for such an OS or ORAM scheme is a function, $T(n)$,
such that the total number of messages 
sent between Alice and Bob during the simulation of all $N$
of her accesses from $\sigma$ is $O(N\cdot T(n))$ with high
probability
(i.e., with probability at least $1-1/n^c$, for some constant $c>1$).
That is, the I/O overhead is the amortized expected number of accesses
to Bob's storage that are done to hide each of Alice's original accesses
in her algorithm, $\mathcal A$.
 Of course,
 we would like $T(n)$ to be as small as possible.

In this paper, we provide methods
for improving the asymptotic I/O overhead
for ORAM simulations by more than just constant factors.
The approach we take to achieve this goal is to first transform 
the original RAM access sequence, $\sigma$, into an intermediate OS sequence,
$\hat \sigma$, which has a restricted structure that we refer to as it
being \emph{isogrammic},
and we then efficiently implement an oblivious storage for
this isogrammic sequence,
$\hat \sigma$, transforming it 
into a final access sequence, $\sigma'$, by taking advantage
of this restricted structure.
We define a sequence, $\sigma=(\sigma_1,\sigma_2,\ldots,\sigma_N)$
of $\Put$ and $\Get$ operations
to be \emph{isogrammic}$\,$\footnote{An \emph{isogram} is a word,
  like ``copyrightable,'' without a repeated letter.
  \hfil\break E.g., see \url{wikipedia.org/wiki/Isogram}.}, 
for an underlying set of size, $n$,
if it satisfies the following conditions:
\begin{itemize}
\item
For every $\Get(k)$ operation, there is a previous $\Put(k,v)$
operation, with the same key, $k$.
\item
No $\Put(k,v)$ operation attempts to add
an item $(k,v)$ when the key $k$ is already in the set.
\item
The key, $k$, used in each $\Get(k)$ (resp., $\Put(k,v)$) operation contains
random component chosen independently and uniformly
of at least $\lceil \log n\rceil$ bits.
Thus, keys are unlikely to be repeated.
\end{itemize}

That is, each time we issue a $\Put(k,v)$ operation, the key, $k$, contains
a random nonce that is chosen independently from any of the previous random nonces
we chose for previous keys.
At a high level, then, our two-phase ORAM simulation scheme is 
surprisingly simple,
in that it combines the classic well-known 
B-tree data structure, which is ubiquitous in database applications,
with isogrammic OS.
That is, at a high level our scheme can be summarized as 
\[
\mbox{\bf B-trees} \ +\ \mbox{\bf Isogrammic-OS}\ \implies\ \mbox{\bf ORAM}.
\]
Thus, we call our scheme \emph{BIOS ORAM}.\footnote{Our scheme implements 
   an ORAM by a reduction to an isogrammic OS, i.e., by replacing a 
   sequence of $\Read$ and $\Write$ operations with a sequence of
   $\Get$ and $\Put$ operations. 
   If one desires a scheme that works entirely in the ORAM framework,
   one can, for example, implement these $\Get$ and $\Put$ operations using
   a cuckoo hash table, with $O(1)$ lookup times in the worst case and 
   $O(1)$ amortized insertion times with high probability (w.h.p.). This will result
   in a sequence of $\Read$ and $\Write$ operations and it does not reveal
   any information about Alice's original access sequence, since the
   $\Get$ and $\Put$ operations come from
   an OS.}
This approach allows us to achieve the main goals of this paper,
which is the design of ORAM simulation methods
that work for a wide range of values to the parameters
$B$ and $M$.
Moreover, we are able to use this approach to design
schemes that are statistically secure
and, w.h.p., have efficient I/O overhead bounds.

\begin{table*}[htb]
\ifFull\footnotesize\fi
\begin{center}
\sffamily
\begin{tabular}{|l|c|c|c|c|c|}
\hline\hline
\rule[-3pt]{0pt}{13pt}
\textbf{Method} & \textbf{Parameterizable?} & \textbf{Statistically Secure?} 
& $B$ & $M$ & \textbf{I/O Overhead} \\
\hline\hline
\rule[-3pt]{0pt}{14pt}
Goldreich-Ostrofsky~\cite{Goldreich:1996} & No & No & $\Theta(1)$ & $\Theta(1)$ &
$O(\log^3 n)$ \\
\hline
\rule[-3pt]{0pt}{14pt}
Kushilevitz {\it et al.}~\cite{Kushilevitz:2012} 
& No & No & $\Theta(1)$ & $\Theta(1)$ 
        & $O(\log^2 n/\log\log n)$ \\
\hline
\rule[-3pt]{0pt}{14pt}
Damg{\aa}rd {\it et al.}~\cite{dmn-11} & No & Yes & $\Theta(1)$ & $\Theta(1)$ &
$O(\log^3 n)$ \\
\hline
\rule[-3pt]{0pt}{14pt}
Supermarket ORAM~\cite{Chung2014} 
& No & Yes 
        & $\Theta(1)$
        & $\Theta({\rm polylog}\ n)$ 
        & $O(\log^2 n\log\log n)$ \\
\hline
\rule[-3pt]{0pt}{14pt}
Goodrich-Mitzenmacher~\cite{Goodrich2011} 
& Somewhat & No
& $\Theta(1)$ & $\Theta(n^{\epsilon})$ & $O(\log n)$ \\
\hline
\rule[-3pt]{0pt}{14pt}
Melbourne shuffle~\cite{Ohrimenko2014}
& Somewhat & No 
& $\Theta(n^\epsilon)$ & $\Theta(n^{\epsilon})$ 
& $O(1)$ \\
\hline
\rule[-3pt]{0pt}{14pt}
Path ORAM~\cite{Stefanov:2013}
& Yes & Yes 
        & $\omega(\log n)$
        & $\omega(B\log n)$ 
        & $O(\log^2 n/\log B)^\dagger$ \\
\hline\hline
\rule[-3pt]{0pt}{14pt}
\textbf{BIOS ORAM} 
        & \textbf{Yes} & \textbf{Yes} & $\Omega(\log n)$ & $\Omega(\log n)$ 
        & $O(\log^2 n/\log^2 B)$ \\
\hline\hline
\end{tabular}
\end{center}
\caption{\label{tbl:results} Our BIOS ORAM bounds (in boldface), 
compared to some of the asymptotically best previous ORAM methods, which
are distinguished between those results 
have parameterized message/memory sizes or not,
and those that are statistically secure or not.
The parameter $0<\epsilon\le 1/2$ is a fixed constant.
The above results that are not statistically secure assume the existence
of random one-way hash functions, i.e., they assume the existence of random
oracles.
\hfil\break
$^\dagger$The Path ORAM method~\cite{Stefanov:2013}
claims an I/O bandwidth of $O(\log n)$
blocks, but this assumes that blocks can contain the responses
of multiple back-and-forth messages;
hence, we use the above bound to characterize
the I/O overhead for Path ORAM, which counts the 
actual number of messages, each of size at most $B$.
}
\end{table*}

\subsection{Previous Related Results}
Work on 
ORAM simulation methods traces its origins to seminal work of
Goldreich and Ostrofsky~\cite{Goldreich:1996},
who achieve an I/O overhead of $O(\log^3 n)$ with $M$ and $B$ being
$O(1)$ using a scheme that fails with polynomial probability and is not
statistically secure.
Kushilevitz {\it et al.}~\cite{Kushilevitz:2012} 
improve the I/O overhead for ORAM with a constant-size client-side
memory to be $O(\log^2 n/\log\log n)$, albeit while still not achieving
statistical security.
Damg{\aa}rd {\it et al.}~\cite{dmn-11} introduce an ORAM scheme that is
statistically secure with an I/O overhead that is $O(\log^3 n)$, with
$M$ and $B$ being $O(1)$, that is, there method is not parameterized
for general values of $M$ and $B$.
Chung {\it et al.}~\cite{Chung2014} provide a statistically secure
ORAM scheme, which we are calling ``supermarket'' ORAM (due to its reliance on
an interesting ``supermarket'' analysis),
for the case when $B$ is $O(1)$ and $M$ is at least polylogarithmic, 
which has an I/O overhead of $O(\log^2 n\log\log n)$.
Unfortunately, these previous schemes do not apply
to parameterized scenarios with larger values for $B$, such as when $B$
is at least logarithmic, let alone for cases when $B$ is $\Omega(n^\epsilon)$,
for some constant $0<\epsilon\le 1$.
Interestingly, Goldreich and Ostrofsky~\cite{Goldreich:1996}
give a lower-bound argument that 
the I/O overhead for an ORAM scheme must be $\Omega(\log n)$ when $M$ is
$O(1)$, but their lower bound does not apply to 
larger values of $M$ and $B$.

There is previous work that is parameterized for larger values of $M$ and $B$,
however.
Goodrich and Mitzenmacher~\cite{Goodrich2011}
provide an ORAM simulation scheme that achieves
an $O(\log n)$ I/O overhead and constant-sized messages, 
but their method requires $M$ to
be $\Omega(n^\epsilon)$, for some fixed constant $0<\epsilon\le 1$,
which is not fully parameterized, e.g., for when $M$
is polylogarithmic.
Also, their method is not statistically secure.
Stefanov {\it et al.}~\cite{Stefanov:2013}
introduce the Path ORAM method,
which is statistically secure and parameterized for 
values of $B$ that are super-logarithmic and values of $M$ that are at least
a logarithmic factor larger than $B$,
to achieve an I/O overhead (in terms of the number of messages
exchanged between Alice and Bob) of $O(\log^2 n/\log B)$.
That is, their method also can match the logarithmic lower bound
of Goldreich and Ostrofsky, but 
with a scheme that requires both $M$ and $B$ to be $\Omega(n^\epsilon)$,
for a fixed constant $0<\epsilon\le 1$.
Ohrimenko {\it et al.}~\cite{Ohrimenko2014} present
an oblivious storage (OS) scheme (which, as we observed, can also be used
for ORAM simulation) that achieves an I/O overhead of
$O(1)$, but it is not statistically secure and it requires 
$M$ and $B$ to be at least $\Omega(n^\epsilon)$.

Wang {\it et al.}~\cite{Wang:2014} introduce an 
interesting ``oblivious data structure'' framework, which 
applies to 
bounded-degree data structures, such as search
trees, to achieve an $O(\log n)$ I/O overhead for data-structure
access sequences.
Their algorithms are based on (non-recursive) 
Path ORAM~\cite{Stefanov:2013}, however, which requires that
$M$ and $B$ be super-logarithmic, and, even then,
their method does not achieve an $O(1)$ I/O overhead,
even for larger values of $B$ and $M$.

In addition, there is a growing
literature on other ORAM and OS solutions, which is too large to review 
here
(e.g., see~\cite{Apon2014,Stefanov:2013:MOS:2508859.2516673,%
Shi2011,Goodrich:2012:POS,Goodrich:2011:ORS:2046660.2046680,%
Goodrich:2012:PGD:2095116.2095130,Ren:2013,s6547114,cryptoeprint:2014:997,%
Gentry2013,Devadas2016,Wang:2015}).
These results also optimize I/O overhead, but we are not aware
of any previous results that beat the asymptotic I/O bounds for the Path ORAM
scheme for a wide range of values of $B$ and $M$ while also achieving
statistical security for the simulation method.

\subsection{Our Results}
We provide a method for ORAM simulation, which we call BIOS ORAM,
that achieves statistical security and has efficient asymptotic I/O overheads
for a wide range of values for the parameters $B$ and $M$.
In particular, we show how to perform an ORAM 
simulation of a polynomial number of accesses to an outsourced storage
of size $n$ with
an I/O overhead that is $O(\log^2 n/\log^2 B)$, w.h.p.,
for $B$ and $M$ ranging from logarithmic to a fraction of $n$.
For example, we can achieve the following specific bounds, depending
on the values of $B$ and $M$:
\begin{itemize}
\item
When $B$ and $M$ are logarithmic or polylogarithmic in $n$,
we achieve an I/O overhead that is $O(\log^2 n/(\log\log n)^2)$, w.h.p.
\item
When $B$ and $M$ are only $\Omega(2^{\sqrt{\log n}})$,
we achieve an I/O overhead that is $O(\log n)$, w.h.p.
\item
When $B$ and $M$ are $O(n^\epsilon)$, for some constant $0<\epsilon\le 1/2$,
we achieve an I/O overhead that is $O(1)$, w.h.p.
\end{itemize}

We summarize our results in Table~\ref{tbl:results}, comparing them to 
some of best-known previous ORAM results.
Note, for example, that
our results apply to a wider range of values of the parameters
$B$ and $M$ than the Path ORAM scheme~\cite{Stefanov:2013} and improves the I/O
overhead for ORAM simulation over this entire range.
For example, the best I/O overhead that Path ORAM can achieve is 
$O(\log n)$ even when $B$ and $M$ are $O(n^\epsilon)$, whereas our BIOS ORAM
scheme achieves a constant I/O overhead in such scenarios.
In addition,
our I/O overhead bounds match those of the Melbourne shuffle for 
values of $B$ that are $\Theta(n^\epsilon)$ while also extending to values
of $B$ that are smaller than those possible using the Melbourne shuffle approach.
For example, as mentioned above,
if $B$ and $M$ are just $\Omega(2^{\sqrt{\log n}})$, then we achieve
an I/O overhead of $O(\log n)$, w.h.p., which is a result that is
not achievable using previous ORAM methods for such values of $B$ and $M$.

Our methods are remarkably simple and make
multiple uses of the ubiquitous B-tree data structure
(e.g., see~\cite{Comer:1979,Cormen:2009:IAT:1614191,Goodrich:2014:ADA}),
along with a reduction of ORAM simulations to isogrammic OS access 
sequences, as well as efficient ways of obliviously simulating
isogrammic access sequences (again, using B-trees).
In particular, we show how to implement an oblivious storage (OS)
scheme for any isogrammic access sequence so as to achieve a simulation
that achieves statistical security and has an I/O overhead that
is $O(\log n/\log B)$ with high probability.
In addition, we show how to apply this result to improve
the I/O overhead for oblivious tree-structured data structures, which improves
an oblivious data-structure bound of 
Wang {\it et al.}~\cite{Wang:2014} and may be of independent interest.

\section{An Overview of B-Trees}
As is well-known in database circles,
a B-tree
is a multi-way search tree, which stores internal nodes as blocks
so that its depth is $O(\log_B n)$,
e.g., see~\cite{Comer:1979,Cormen:2009:IAT:1614191,Goodrich:2014:ADA}.
In the B-trees we use in this paper, we choose a branching factor 
of 
\[
B'=B^{1/4},
\]
where $B$ is 
our message-size parameter.
Such a B-tree 
supports searching and updates (insertions and deletions) in 
$O(\log_{B'} n)=O(\log n/\log B')=O(\log n/\log B)$
I/Os of blocks of size~$B'$.
That is, each search or update involves accessing $O(1)$ nodes on each level
of the B-tree, in a root-to-leaf search followed (for updates)
by a leaf-to-root set of updates,
for which we refer the interested reader to known
methods for searching and updating B-trees
(e.g.,
see~\cite{Comer:1979,Cormen:2009:IAT:1614191,Goodrich:2014:ADA}).
%
From the perspective of the server, Bob, the I/Os for 
searching a B-tree would simply
look like Alice accessing $O(\log n/\log B')=O(\log n/\log B)$ 
blocks of storage.
(See Figure~\ref{fig:b-tree}.)

\begin{figure}[hbt!]
\begin{center}
\includegraphics[trim = 1.5in 2.1in 2.3in 3.4in, clip, width=3.4in]{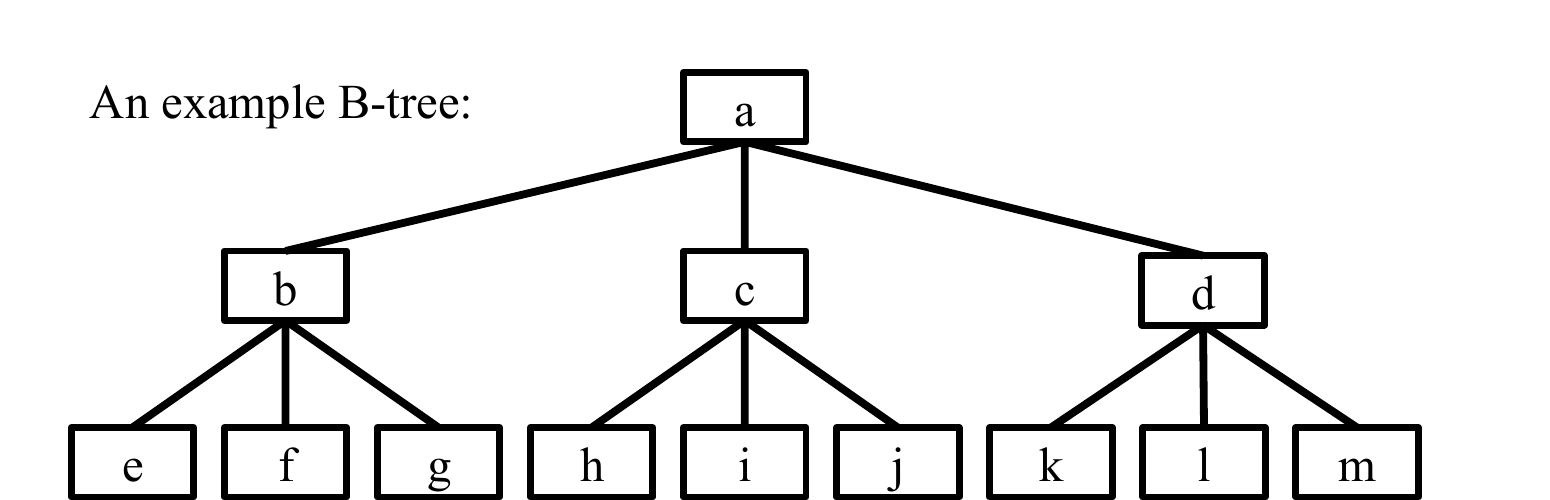}
\end{center}
\caption{\label{fig:b-tree} An example B-tree with branching factor 3.
\copyright~2017 Michael Goodrich. Used with permission.
}
\end{figure}

\section{B-Trees + Isogrammic OS = ORAM}
In this section, we describe the first component of
our \emph{BIOS} ORAM scheme, which is a reduction of ORAM simulation
to isogrammic OS, at the cost of increasing the I/O overhead of Alice's
accesses by a factor of
$O(\log n/\log B)$. 
That is, we show how to transform an arbitrary
sequence of $\Read$ and $\Write$ operations
into an isogrammic sequence of $\Get$ and $\Put$ operations, with a blow-up
in length of $O(\log n/\log B)$.
By then showing how to obliviously simulate an isogrammic access sequence with 
an I/O overhead of 
$O(\log n/\log B)$, we get the main result of this paper, 
that is, that we can achieve an ORAM scheme with an I/O overhead of
$O((\log n/\log B)^2)$.

Suppose, then, that Alice's RAM algorithm, $\mathcal A$,
which she wishes to perform on her data outsourced to Bob, uses a memory
of $n$ cells indexed by integers in the range $[0,n-1]$.
Let $R$ be a B-tree having each cell of Alice's storage 
stored in a sub-block of size $B'=B^{1/4}$ associated with a block
at a leaf of $R$ 
(ordered in the standard left-to-right fashion).
Furthermore, let $R$ have a branching
factor of $B'=B^{1/4}$, so each internal node in
$R$ can be stored in a single sub-block of size $B'$ and the depth
of $R$ is $O(\log n/\log B')=O(\log n/\log B)$.
Intuitively, the main idea of our reduction is that,
for each $\Write(i,v)$ or $\Read(i)$ operation in $\mathcal A$, 
we perform a search in
$R$ for the index $i$, to find the sub-block containing memory cell, $i$,
and then we replace this sub-block of size $B'$ and 
all the nodes of size $B'$ that we just traversed with new nodes.

Initially, we construct $R$ in a bottom-up fashion so that each node $u$ in $R$
is assigned a random nonce, $r_u$, of $\lceil\log n\rceil$ bits.
For each leaf, $u$, of $R$, which is storing some block, $V$, of $B$ 
values for the cells in Alice's storage
for some set of indices, $\{i,i+1,\ldots,i+{B'-1}\}$, 
we issue a $\Put(k,v)$ operation, where $k=(r_u,u)$ and
$v=(i,V)$.
In addition, we (obliviously) store $r_u$ at $u$ in $R$ on the server.
Note that $v$ fits in a single block of size $B$.
For each internal node, $u$, of $R$, which we construct level-by-level, so that we
can obliviously read in the random nonce, $r_{u_i}$,
 associated with each child, $u_i$, of $u$ in $R$,
then we assign $u$ a random nonce, $r_u$, and we issue a $\Put(k,v)$ operation,
where $k=(r_u,u)$ and
\[
v=(I,r_{u_1},u_1,r_{u_2},u_2,\ldots,r_{u_{B'-1}},u_{B'-1}),
\]
where $I$ is the block of key values needed to decide for any search which
child, $u_i$, of $u$ to access next.
Note that $v$ fits in $O(1)$ sub-blocks of size $B'$.
This initialization phase establishes 
the B-tree, $R$, in Bob's storage and issues
$O(n/B')$ $\Put$ operations that identify each node, $u$, of $R$ using a key
that comprises a random nonce, $r_u$, of $\lceil\log n\rceil$ bits.
We also keep a global variable, that maintains the random nonce for
the root.

For each $\Read(i)$ or $\Write(i,v)$ operation after this
initialization, we traverse a root-to-leaf path, $\pi$, in $R$ to the
leaf associated with the sub-block holding the cell $i$. 
This involves performing a sequence of $O(\log n/\log B)$ 
$\Get(k)$ operations, where each $k$ is a pair of a node name in $R$
and the random nonce for that node.
We cache the nodes returned by these $\Get$ operations 
in Alice's private memory.
Then, processing the nodes in $\pi$ in reverse order, we
give each node $u$ in this sequence a new random nonce, $r_u$,
and we issue a $\Put(k,v)$ operation,
where $k=(r_u,u)$ and
\[
v=(I,r_{u_1},u_1,r_{u_2},u_2,\ldots,r_{u_{B'-1}},u_{B'-1}),
\]
where $u_1,\ldots,u_{B'-1}$ are the children of $u$ in $R$.
We do this processing in reverse order so that when we issue
such a $\Put(k,v)$ operation, we will have available the 
new random nonce for the child, $u_j$, of $u$ that we previously
processed as well as the old (and still unused) random
nonces for the other children of $u$.
In this case, the nodes are B-tree nodes, which are of size $B^{1/4}$,
as we have defined this to be the branching factor for our B-tree.
Fortunately, our message block size, $B$, is large enough to store up
to $B^{3/4}$ such nodes in a single message block.

Each $\Read$ or $\Write$ 
involves issuing $O(\log n/\log B)$
$\Get$ and $\Put$ operations; hence, 
this increases the total number of I/Os for Alice's access sequence
by an $O(\log n/\log B)$ factor.
The important observation is that this process results
an isogrammic access sequence, 
since each key used in a $\Put(k,v)$ operation comprises a random
nonce of at least $\lceil\log n\rceil$ bits, each such key is not already in our
set (since each key also comprises a unique node name),
and each $\Get(k)$ operation is guaranteed to 
match up with a previous $\Put(k,v)$ operation.
Thus, we have the following.

\begin{theorem}
\label{thm:isogrammic-b-tree}
Given a RAM algorithm, $\mathcal A$,
with memory size, $n$,
we can simulate the memory accesses of $\mathcal A$ using an isogrammic access
sequence that initially creates $O(n/B')$ $\Put$ operations and then
creates $O(\log n/\log B)$ $\Get$ and $\Put$ operations 
for each step of $\mathcal A$.
Each key used in a $\Get$ or $\Put$ operation comprises a random
nonce of at least $\lceil \log n\rceil$ bits and each value used in a $\Put$
operation is a sub-block of size $O(B')$ words.
\end{theorem}

\begin{figure*}[hbt!]
\begin{center}
\includegraphics[trim = 0.5in 1.5in 0.5in 1.7in, clip, width=4.5in]{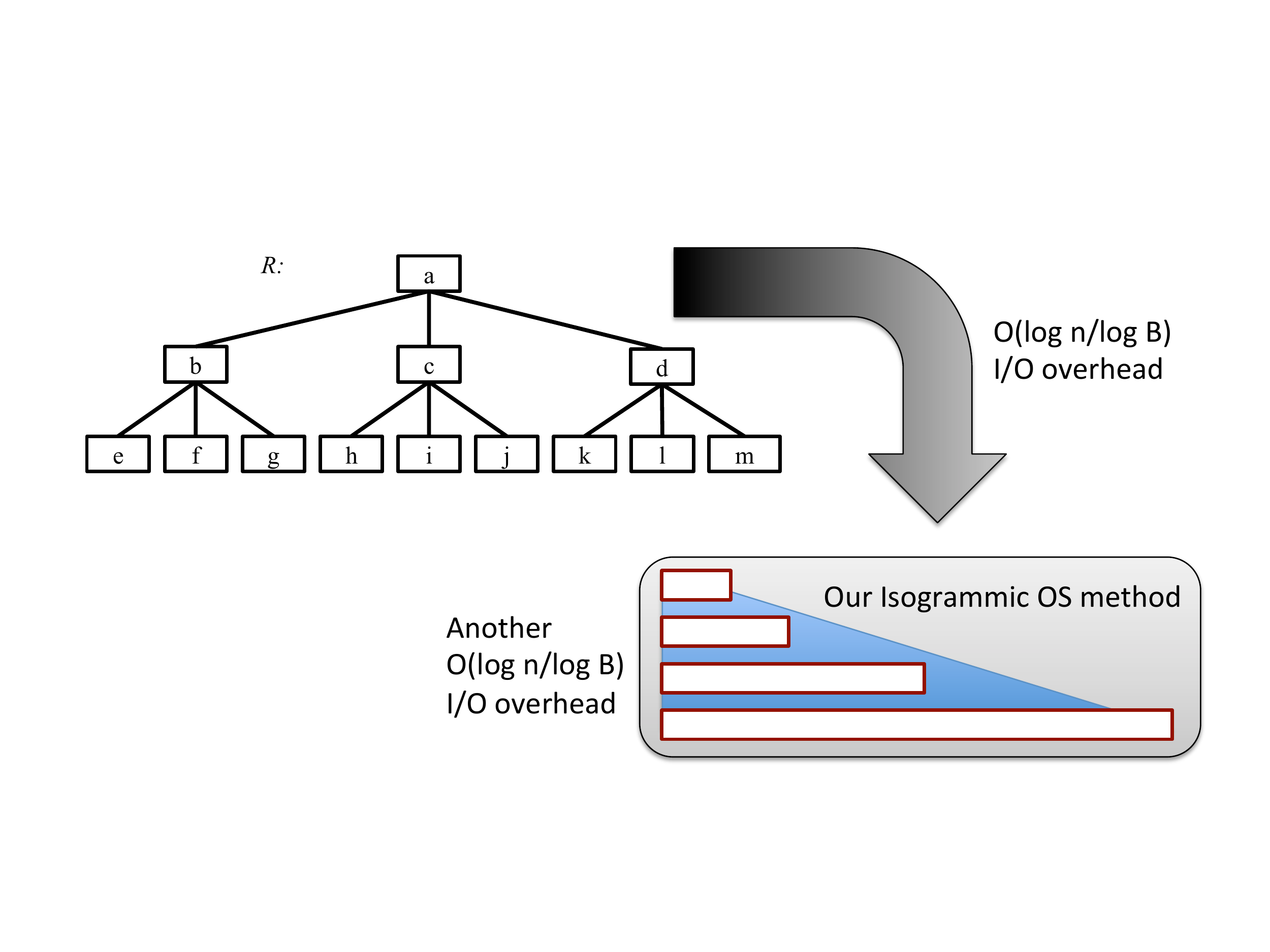}
\end{center}
\vspace*{-10pt}
\caption{\label{fig:top-level} A high-level view of 
our BIOS ORAM scheme.
\copyright~2017 Michael Goodrich. Used with permission.
}
\end{figure*}

\begin{proof}
For the security claim,
consider a simulation of the security game
mentioned in the introduction, assuming the statistical security for
our isogrammic OS.
Suppose, then, that
Bob creates two
access sequences, $\sigma_1$ and $\sigma_2$, and gives them to Alice,
who then chooses one at random and simulates it, as described above.
For each access to a memory index, $i$, in the RAM simulation for her
chosen $\sigma_j$,
the memory cell for $i$ is read and written to by doing a search in $R$.
The important observation is that this access consists
of $O(\log n/\log B)$ accesses a root-to-leaf sequence of nodes of $R$, indexed
by newly-generated independent random numbers each time.
Thus, nothing is revealed to Bob about the index, $i$.
That is, the number of accesses in Alice's simulation is the same
for $\sigma_1$ and $\sigma_2$, and the sequence of keys used is completely 
independent of the choice of $\sigma_1$ or $\sigma_2$.
Thus, Bob is not able to determine which of these sequences she chose
with probability better than $1/2$.
\end{proof}

As we show in the remainder of this paper, we can simulate an 
isogrammic access sequence obliviously in a statistically secure
manner with an I/O overhead that is 
$O(\log n/\log B)$ with high probability.
Combining this result with 
Theorem~\ref{thm:isogrammic-b-tree} gives us our claimed result 
that we can perform statistically-secure ORAM simulation with
an I/O overhead that,
with high probability,
is $O(\log^2 n/\log^2 B)$.
(See Figure~\ref{fig:top-level}.)

\section{B-tree OS for Small Sets}
\label{sec:b-tree}
In this section, we present an OS solution for small sets,
that is, sets whose size, $n$, is $O(B^{3/2})$, and small items,
that is, items whose size is $O(B')$ words, where $B'=B^{1/4}$.
This is admittedly a fairly restrictive scenario, but it 
nevertheless is a critical component of our isogrammic OS scheme.
We show how, in this scenario, to use a B-tree
to achieve statistical security for an OS simulation that works for
general sequences of $\Put$ and $\Get$ operations, not just
isogrammic sequences.
That is, in this subsection, we allow for keys that are not
necessarily random (so long as they are unique) and we allow
$\Get(k)$ operations to return ``not found'' responses.

Suppose, then, that we wish to support an oblivious storage (OS) for
a set of items that can be as large as $n=\Theta(B^{3/2})$.
In this case, we utilize a B-tree, $F$, with branching factor, 
$B'=B^{1/4}$, so its height is 
$O(\log n/\log B')=O(\log n/\log B)=O(1)$, since we are restricting
ourselves here to sets of size at most $O(B^{3/2})$ and items of size
$O(B')$.
Put another way, our restriction on the set size, $n$, implies that
$B$ is $\Omega(n^{2/3})$ and $B'$ is $\Omega(n^{1/6})$,
that is, that $B$ is $\Omega(B^{1/4}n^{1/2})$.

Our small-set B-tree OS method is a
modification and adaptation of the ``$\sqrt{n}\,$'' solution 
of Damg{\aa}rd {\it et al.}~\cite{dmn-11} to B-trees and the OS setting.
Let $F$ be our B-tree with capacity $n$ and branching factor $B'$;
hence, $F$ has depth $D=4\lceil\log n/\log B\rceil$,
with $n$ items stored in its leaves, and randomly shuffled in 
an array of Bob's storage of size $O(nB')$ memory words. 
Initially, we construct $F$ using an initial set of items, which could even
be empty. We pad the nodes of $F$ with empty dummy nodes, as necessary,
to make every node of $F$ have the same depth, $D$.
An internal B-tree node consists of $B'$ keys; hence, a single
message block can fit $B^{3/4}$ B-tree nodes or items (since items
are also of size $B'$ here), that is,
$\Omega(n^{1/2})$ B-tree nodes or items.
In addition, we also store a singly linked list, $\ell$, 
of $D\lceil\sqrt{n}\rceil$ nodes, which are the
same size as B-tree nodes and items and are randomly shuffled in with the 
B-tree nodes of $F$.

We can initialize $F$ in this way using the oblivious shuffling
method of Goodrich and Mitzenmacher~\cite{Goodrich2011}.
In this context, where we are obliviously sorting nodes and items
that are themselves of size $B'$,
their method involves a multi-way merging of sorted lists
with a branching factor of $(M/B')^{1/3}=\Theta(n^{1/6})$,
where the merging step involves 
a simple oblivious scanning of each of these arrays.
The merge in their method requires, in this context (where we are
merging nodes and items of size $B'=\Theta(n^{1/6})$) that there be 
\[
\Omega((M/B')^{2/3}B' + (M/B')^{1/3}(B')^2)
\]
values stored in memory at any given time,
which we can bound as $\Omega(n^{1/2})$ using the fact that, in this case,
\[
(M/B')^{2/3}B' + (M/B')^{1/3}(B')^2
\]
is at least
$
(n^{1/2})^{2/3}n^{1/6} + (n^{1/2})^{1/3}n^{1/3}$,
which equals $2n^{1/2}$.
Thus, we can scan each of the $\Theta(n^{1/6})$ sorted lists by
reading $\Theta(n^{1/2})$ elements from each sorted list at a time
(and stopping the recursion when we reach lists of this size).
This means that the total number of I/Os needed for this oblivious
shuffling is $O((n/n^{1/2})\log^2_{M/B'} (n/B'))$ I/Os.
Given our other assumptions about $B$ and $M$, this shuffling
therefore requires at most $O(n^{1/2})$ I/Os.
In addition, we maintain $D$ caches, $C_1,\ldots,C_D$,
of size $B'\sqrt{n}$ each, one for each level of $F$.
Thus, each cache can be read or written in $O(1)$ I/Os.

\begin{figure*}[hbt!]
\begin{center}
\includegraphics[trim = 1.2in 2.8in 0.8in 2.3in, clip, width=5.5in]{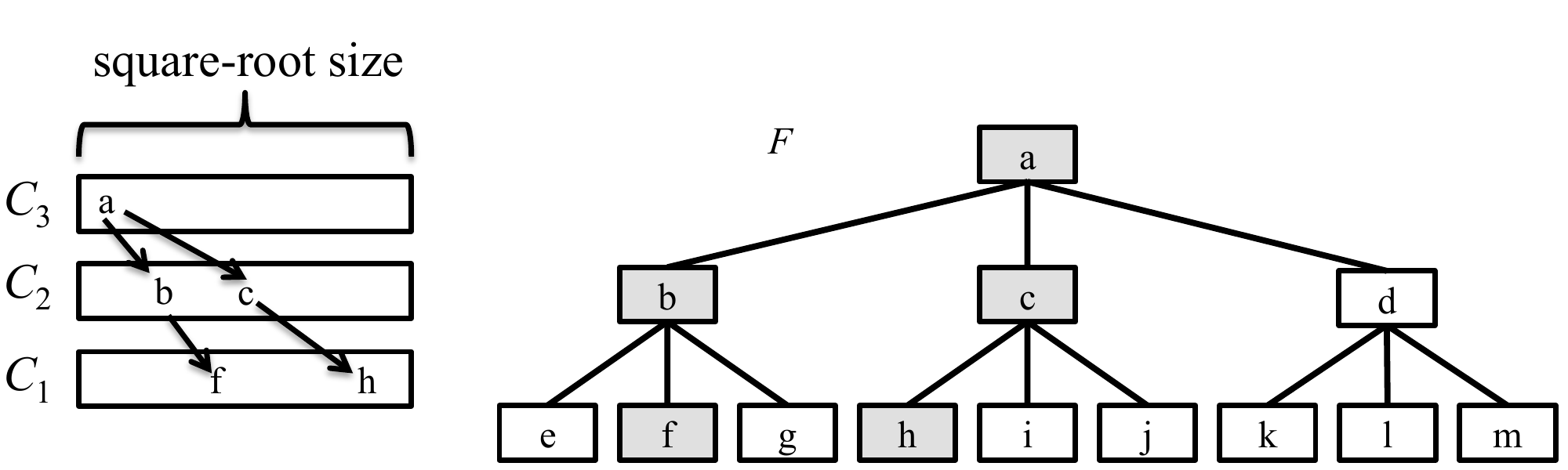}
\end{center}
\caption{\label{fig:fusion2} A B-tree, $F$, 
and a hierarchy of $\Theta(\log n/\log B)$ caches, $C_1,C_2\ldots$ . We
shade the nodes that have been accessed previously in grey. Each cache
is associated with a specific level of $F$.
\copyright~2017 Michael Goodrich. Used with permission.
}
\end{figure*}

Let us consider each type of access, with the design of making it impossible
for Bob to determine even if we are performing a $\Get$ or $\Put$ operation.
In either a $\Get(k)$ or $\Put(k,v)$ operation, we begin with a search for
the key $k$ in $F$.
To perform such a search in $F$, 
we read each of the nodes in a path, $\pi$,
from the root of $F$ to the leaf containing our search key, $k$, or its
predecessor (if $k$ is not in our set, $S$).
For each level, $i$, of $F$ during this search, we first read (as one I/O),
the cache, $C_i$, to see if the $i$-th node, $v_i$, of $\pi$ is in $C_i$.
If $v_i$ is in $C_i$, then we examine it and determine the location of the
next node, $v_{i+1}$, in $\pi$, and we read the next dummy node in $\ell$ (for
the sake of obliviousness, since the location of this dummy node looks
random to Bob and has not been previously accessed).
If $v_i$ is not in $C_i$, then we read it in (this is the first time we are
accessing $v_i$ and this location looks to Bob to be random).
We note that each such node is of size $B'$;  hence, it 
and each cache can be read or written in $O(1)$ I/Os.
(See Figure~\ref{fig:fusion2}.)

After we have completed the reading of 
all the nodes in $\pi$, 
we can perform any updating as necessary for these nodes so as to perform
the functionality of our $\Get(k)$ or $\Put(k,v)$ operation.
Without going into details
(e.g., see~\cite{Comer:1979,Cormen:2009:IAT:1614191,Goodrich:2014:ADA}),
either of these operations will either involve no structural changes
to $F$ or will involve our adding $O(1)$ nodes per level of $F$.
Let $\pi'$ denote the set of updated nodes in $F$ (which we can determine
using Alice's private memory).
Note that each internal node of $F$ in $\pi'$ will including pointers to
existing nodes in $F$, but these have not yet been accessed yet and we have
not revealed any information about them to Bob.
We then write the nodes of $\pi'$ out
to Bob's storage, placing each node, $v_i$, on level of $i$ of $\pi$,
in the cache, $C_i$,
using $O(1)$ I/Os for each level of $F$.
In fact, we pad this set so that we always write the same number of $O(1)$
nodes to each $C_i$, based on standard update rules
for B-trees
(e.g., see~\cite{Comer:1979,Cormen:2009:IAT:1614191,Goodrich:2014:ADA}).
We perform this process in a bottom-up leaf-to-root fastion, so that we can
inductively always be able to know the locations for the child nodes
for any node in $F$ (even if that node is in a cache and some of its children
are in the lower-level cache). 
Thus, we can determine the locations in
Bob's storage for any root-to-leaf path in $F$ by reading the nodes
and caches in Bob's storage
sequentially starting from the root.
After we have completed $\lceil\sqrt{n}\rceil$ accesses of $F$ in this
manner, we rebuild and reshuffle $F$ and a new linked list, $\ell$, and repeat
this B-tree access procedure.
This rebuilding and reshuffling requires $O({n}^{1/2})$ I/Os,
as described above.
Thus, we have the following.

\begin{theorem}
\label{thm:sqrt}
Suppose we have a set, $S$, of up to $n$ items, 
where each item is of size at most $B'$,
where $B'=B^{1/4}$, and $n$ is $O(B^{3/2})$.
Then our B-tree OS solution can implement an 
oblivious storage for $S$ that has an I/O overhead of 
$O(\log n/\log B)=O(1)$, with high probability.
This simulation is statistically secure, even for non-isogrammic access
sequences.
\end{theorem}
\begin{proof}
With respect to the security of this method, note that each access
involves a search of all the caches of size $\sqrt{n}$ and an access to a
distinct random location (which is either a real node or a dummy node 
that Bob cannot tell apart)
for each level of a shuffled B-tree, $F$,
which was shuffled obliviously and has every possible permutation of
its nodes on that level as being equally likely.
We then access each cache for every level of $F$ in a bottom-up fashion.
Moreover, both the top-down and bottom-up phases of this computation involve
the same form of access irrespective of whether we are performing
a $\Get$ or $\Put$ operation.
Thus, the adversary, Bob, can learn nothing about Alice's access sequence
based on observing her access pattern.
That is, in terms of the security game, Bob is unable to
distinguish between two access sequences, $\sigma_1$ and $\sigma_2$,
of length $N$ for sets of up to $n$ items.

Since $D$ is $O(\log n/\log B)$, 
and $B$ is large enough for us to read an entire cache with one I/O,
it is easy to see that each access in this simulation requires
$O(\log n/\log B)=O(1)$ I/Os.
In addition, after we have performed $O(\sqrt{n})$ such accesses,
we do a rebuilding action that requires $O(n^{1/2})$ I/Os;
hence, this adds an amortized $O(1)$ 
number of I/Os for each of the previous
$O(\sqrt{n})$ accesses.
Thus, the total I/O overhead is $O(\log n/\log B)=O(1)$, with high
probability (where this probability is dependent only on the algorithm
we use for oblivious shuffling, e.g., see~\cite{Goodrich2011}).
\end{proof}

\ifFull
This solution is a crucial component of our general isogrammic OS
solution, which we describe next.
\fi

\section{Isogrammic Oblivious Storage}
\label{sec:isogrammic-os}
In this section, we describe our isogrammic OS scheme, which is able
to obfuscate any isogrammic access sequence with statistical security,
achieving an I/O overhead that is $O(\log n/\log B)$ with high
probability.
Our construction involves yet another use of B-trees, as a primary
search structure, as well as repeated uses of our B-tree OS for small
sets from Section~\ref{sec:b-tree}.

\begin{figure*}[htb!]
\begin{center}
\includegraphics[trim = 0.4in 2.7in 3.4in 1.8in, clip, width=4.2in]{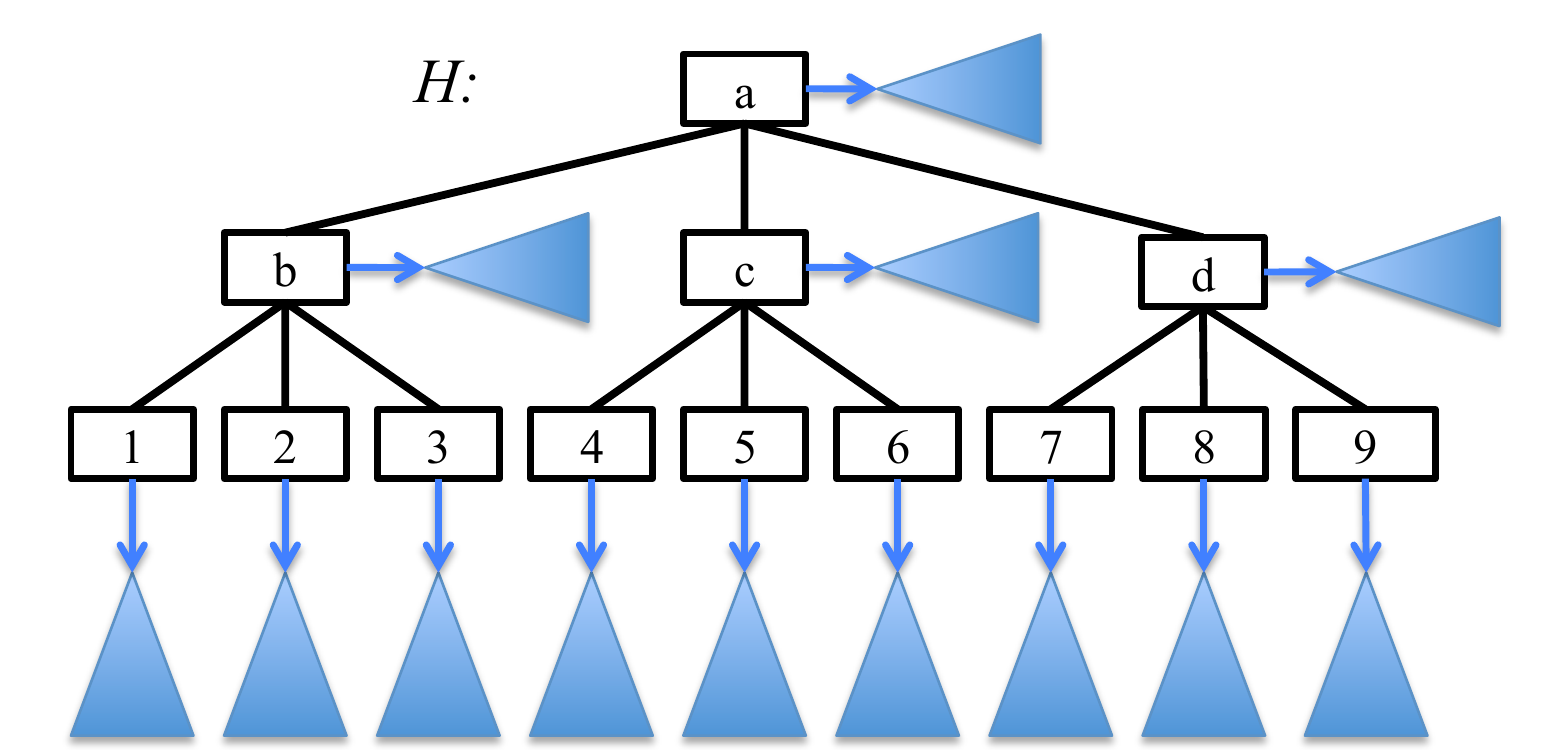}
\end{center}
\vspace*{-4pt}
\caption{\label{fig:isogrammic} An illustration of our isogrammic OS.
The B-tree tree, $H$, is shown in black. Each bucket, which
implements our B-tree OS for small sets, is shown as a blue triangle.
That is, each triangle represents a bucket of up to $B^{3/2}$ items,
which are accessed according to our B-tree OS method of 
Section~\ref{sec:b-tree}.
\copyright~2017 Michael Goodrich. Used with permission.
}
\end{figure*}

Let $n$ be the size of a set for which we wish to support
an oblivious storage (OS),
and let
$H$ be a static B-tree of height $O(\log n/\log B)$ with branching factor
$B'=B^{1/4}$, such that $H$ has $n/B$ leaves.
For every node, $u$ in $H$, including both the internal nodes and leaves,
we store a ``bucket,'' $b_u$, which maintains an instance of our B-tree OS,
as described above in Section~\ref{sec:b-tree},
and let each such bucket have capacity $4L$, where $L=B^{3/2}$,
except for leaves, which each have
capacity $8L$.
These buckets are used to store $(k,v)$ items that are in the current set,
i.e., items for which we have processed a $\Put(k,v)$ operation and have yet
to perform a $\Get(k)$ operation.
See Figure~\ref{fig:isogrammic}.

Recall that in an isogrammic access sequence we are given a sequence of
$\Put(k,v)$ and
$\Get(k)$ such that $\Get(k)$ operations
always have an item to return (i.e., there is a previous matching 
$\Put(k,v)$ operation) and $\Put(k,v)$ operations never try to insert an item
whose key matches the key of an existing item.
More importantly, every key contains a random nonce component of 
at least $\lceil \log n\rceil$ bits.
We use these random nonces as the addresses for where items should go
in $H$. Namely, we maintain the following invariant throughout our OS
simulation:
\begin{itemize}
\item
For each item, $(k,v)$, in our current set of items, $(k,v)$
is stored in exactly one bucket, $b_u$, for a node, $u$, on the root-to-leaf
search path in $H$ for the random part of $k$.
\end{itemize}
Given this invariant, let us describe how we process $\Put$ and $\Get$
operations.

For a $\Put(k,v)$ operation, we add $(k,v)$ to the bucket, $b_r$, 
for the root, $r$, of $H$,
using the B-tree OS method described in
Section~\ref{sec:b-tree}.
Note that this satisfies our invariant for storing items in $H$,
that is, storing an item in bucket for the root implies that it is
stored in the root-to-leaf search path for the random part of its key.
(We will describe later what we do when the root bucket becomes full,
but that too will satisfy our invariant.)
Then, for the sake of obliviousness (so Bob cannot tell whether
this operation is a $\Get$ or $\Put$),
we uniformly and independently choose a random key, $k'$, 
and traverse the root-to-leaf path in
$H$ for $k'$, performing a search for $k'$ in the bucket, $b_u$, for
each node $u$ on this path, using
the B-tree OS method described in
Section~\ref{sec:b-tree}.
Alice just ``throws away'' the results of these searches, but, of
course, Bob doesn't know this.

For any given $\Get(k)$ operation,
we begin, for the sake of obliviousness, by inserting a 
dummy item, $(k',e)$, in the bucket, $b_r$, for the root, $r$, 
of $H$, where $e$ is a special ``empty'' value (that nevertheless 
has the same size as any other value) and $k'$ is a random key,
using the fusion-tree OS method described in
Section~\ref{sec:b-tree}.
So as to distinguish this type of dummy item from others, we refer to each
such dummy item as an \emph{original} dummy item.
We then traverse the root-to-leaf path, $\pi$, for (the random part of) $k$ in 
$H$, and, for each node, $u$, in $\pi$, we search in the bucket,
$b_u$, for $u$,
to see if the key-vaue pair for $k$ is in this bucket, using
the B-tree OS scheme described above in
Section~\ref{sec:b-tree}.
By our invariant, the item, $(k,v)$, must be stored in the bucket
for one of the nodes in the path $\pi$.
Note that we search 
in the bucket for every node in $\pi$, even after we have found and removed the 
key-value pair, $(k,v)$. 
Because we are simulating an isogrammic access sequence,
there will be one bucket with this item, but we search all the
buckets for the sake of obliviousness.

An important consequence of the above methods
and the fact that we are simulating an isogrammic access sequence is
that each traversal of a path in $H$ is determined by a random nonce 
that is chosen uniformly at random and is independent of every other
nonce used to do a search in $H$.
Thus, the server, Bob, learns nothing about Alice's access pattern from
these searches.
In addition,
as we will see shortly, 
the server cannot determine where any item,
$(k,v)$, is actually stored, because the random part of the
key $k$ is only revealed when we do a $\Get(k)$ operation and $\Put$
operations never reveal the locations of their keys.
Moreover, we maintain the fact that the server doesn't know the actual
location of any item,
along with our invariant,
even as bucket for a node, $u$, becomes full and needs to 
have its items
distributed to its children.

Periodically, so as to avoid overflowing buckets, 
we move items from a bucket, $b_u$, stored 
at a node $u$ in $H$ to $u$'s children, in a process we call
a \emph{flush} operation.
In particular, we flush the root node, $r$, every $L$ 
$\Put$ or $\Get$ operations.
We flush each internal node, $u$, after $u$ has received $B'$
flushes from its parent, which each involve inserting exactly $4L/B'$ real and
dummy items (including new dummy items) into the bucket for $u$.
Because of this functionality, and the fact that we are moving items
based on random keys, the number of real and
original dummy items in the bucket, $b_u$, at a time when we are flusing
a node $u$ at depth $i$ is expected to be $L$, and we maintain
it to be at most $4L$.
Also, note that
we will periodically perform flush operations across all the
nodes on a given level of $H$ at any given time when flush
operations occur, which is the main  reason why our I/O overhead bounds are
amortized.
We don't flush the leaf nodes in $H$, however.
Instead,
after every leaf, $u$, in $H$ has received $B'$ flushes, we perform an
oblivious compression to compress out a sufficient number of dummy items so that
the number of real and dummy items in $u$'s bucket is $4L$.
Thus, the bucket for a leaf never grows to have more than $8L$ real and dummy
items.
If, at the time we are compressing the contents of a leaf bucket,
we determine that there are more than $4L$ real items being stored in such a
bucket, which, as we show, is an event that occurs with low probability, 
then we 
restart the entire OS simulation. 
Such an event doesn't compromise privacy, since
it depends only on random keys, not Alice's data or access sequence. 
Thus, 
doing a restart just impacts performance, 
but because restarts are so improbable,
our I/O bounds still hold with high probability.

At a high-level, our method for doing a flush operation at a node, $u$, in $H$
has a similar structure to an analogous operation in the 
Path ORAM scheme~\cite{Stefanov:2013}, as well as in the paper mentioned above that
is currently under submission for ORAM simulation when $B$ and $M$ are
both very small.
The details for our flush operation here are different than both of these
works, however,
in that our flush method depends crucially on the B-tree OS method of 
Section~\ref{sec:b-tree}.
\begin{enumerate}
\item
We obliviously shuffle the real and original dummy items
of $b_u$ into an array, $A$, of size $4L$, stored at the server.
This step will never overflow $A$ (because of how we perform the rest
of the steps in a flush operation).
This step can be done using known oblivious 
shuffling methods
(e.g., see~\cite{Goodrich2011}), which add just a constant I/O
overhead factor.
\item
For each child, $x_i$, $i=1,2,\ldots,B'$,
of $u$, we create an array, $A_i$, of size $4L/B'$.
\item
We obliviously sort the real and original dummy  
items from $A$ into the arrays, 
$A_1,\ldots,A_\ell$, according the keys for these items, so that the
item, $(k,v)$, goes to the array $A_i$ if the next $O(\log B')$ bits
of the random part of key $k$ would direct a search for $k$ to the child $x_i$.
We perform this oblivious sorting step so that if there are fewer than 
$4L/B'$ items destined for any array, $A_i$, we 
pad the array with (new) dummy items to bring the number of items destined
to each array, $A_i$, to be exactly $4L/B'$.
However, if we determine from this oblivious sorting step
that there are more than $4L/B'$ real and original dummy 
items destined for any array, $A_i$,
which (as we show) is an event that occurs with low probability, then
we restart the entire OS simulation.
Because this step is done obliviously and search keys are random (hence, 
they never depend on
Alice's data values or access pattern), even if we restart,
Bob learns nothing about Alice's access sequence during this step.
So, let us assume that we don't restart.
This step can be done using known oblivious sorting, padding, and
partitioning methods 
(e.g., see~\cite{Goodrich2011}), which add only a constant I/O
overhead factor.
\item
For each real and dummy item (including both original and new dummy
items), $(k,v)$, in each $A_i$,  we insert 
$(k,v)$ into the bucket $b_{x_i}$ using the B-tree OS method of
Section~\ref{sec:b-tree}. 
This method works only for small sets, but, of course, the number of items
in each bucket determines such a small set.
\end{enumerate}

The first important thing to note about a flush operation is that it
is guaranteed to preserve our invariant that each item, $(k,v)$, is
stored in the bucket of a node in $H$ on the root-to-leaf path
determined by the random part of $k$.
Moreover, because we move real and original dummy items to children
nodes obliviously, in spite of our invariant, the server never
knows where an item, $(k,v)$, is stored; hence, the server can never
differentiate two access sequences more than at random.

Let us analyze the complexity of a flush operation.
Since we flush the root every $L$ steps, and we flush every other
node, $u$, at depth $i$, 
after it has received $B'$ flushes, and both real and original
dummy items are mapped to $u$ only if the first $i\log B'$ bits 
of each of their random keys matches $u$'s  address, 
the expected number of real and original dummy items stored in the
bucket for $u$ is at most $L$ at the time we flush $u$.
In fact, this is a rather conservative estimate, since it assumes
that none of these items were removed as a result of $\Get$
operations.
More importantly, we have the following.

\begin{lemma}
The number, $f$, of real and original dummy items flushed from a
node, $u$, to one of its children, $x_i$, is never more than $4L/B'$, with 
high probability.
Likewise, a leaf in $H$ will never receive 
more than $4L$ real items, with high probability.
\end{lemma}
\begin{proof}
The expected value of $f$, which can be expressed as a sum of
independent indicator
random variables, is at most 
\[
L/B'=B^{3/2-1/4}=B^{5/4}\ge d\log^{5/4} n,
\]
for a constant, $d\ge 3$, since we are assuming that $B$ is
$\Omega(\log n)$.
Thus, by a
Chernoff bound (e.g., see~\cite{mitzenmacher2005probability}),
\[
\Pr(f\ge 4L/B') \le e^{-L/B} \le e^{-d\log^{5/4} n} \le n^{-3\log^{1/4} n}.
\]
The probability bound argument for a leaf in $H$ is similar.
The lemma follows, then, by a union bound across all nodes of $H$ and
the polynomial length of access sequences.
\end{proof}

Thus, with high probability, we never need to do a restart as a
result of a potential overflow during a flush operation.

\begin{theorem}
\label{thm:iso-fusion}
We can obliviously simulate an isogrammic sequence of a polynomial number of
$\Put(k,v)$ and $\Get(k)$ operations,
for a data set of size $n$,
with an I/O overhead of $O(\log n/\log B)$,
with high probability.
Moreover,
this simulation is statistically secure.
\end{theorem}
\begin{proof}
The height of the tree, $H$, is $O(\log n/\log B)$.
Thus, by Theorem~\ref{thm:sqrt}, with high probability,
the I/O overhead is proportional to a constant times $O(\log n/\log B)$,
which is itself $O(\log n/\log B)$.
For the security claim,
consider an instance of the simulation game, where Bob chooses two
isogrammic access sequences, $\sigma_1$ and $\sigma_2$, of length $N$
for a key set of size $n$, and gives them to Alice, who then chooses
one uniformly at random and simulates it according to the isogrammic
OS scheme.
Each access that she does involves accessing a sequence of nodes of $H$
determined by random
keys and for each node doing a lookup in an OS scheme that is 
itself statistically secure, by Theorem~\ref{thm:sqrt}.
In addition, $\Put$ operations add items at the top bucket and are
obfuscated with data-oblivious flush operations.
Therefore, Bob is not able to distinguish between $\sigma_1$ and $\sigma_2$
any better than at random.
\end{proof}

\section{Our BIOS ORAM Algorithm}
Putting the above pieces together, then, gives us the following theorem,
which is the main result of this paper.

\begin{theorem}
\label{thm:oram}
Given a RAM algorithm, $\mathcal A$, with memory size, $n$,
where $n$ is a power of 2,
we can simulate the memory accesses of $\mathcal A$ in an oblivious fashion
that achieves statistical security, such that, with high probability,
the I/O overhead is $O(\log^2 n/\log^2 B)$ for a client-side
private memory of size $M\ge B$ and
messages of size $B\ge 3\log n$.
\end{theorem}
\begin{proof}
By Theorem~\ref{thm:isogrammic-b-tree},
each access
in $\mathcal A$ gets expanded into $O(\log n/\log B)$ operations in an isogrammic
access sequence, and, 
with high probability, each such operation 
has an overhead of $O(\log n/\log B)$,
by Theorem~\ref{thm:iso-fusion}.
The security claim follows from the security claims of 
Theorems~\ref{thm:isogrammic-b-tree}
and~\ref{thm:iso-fusion}.
\end{proof}

\section{Isogrammic Algorithm Design}
\label{sec:applications}

In this section, we study the expressive power of the isogrammic 
access sequences, showing that it subsumes some previous
specialized design patterns for implementing algorithms in the cloud in a
privacy-preserving way.
Thus, 
by Theorem~\ref{thm:iso-fusion},
any algorithm
designed in this framework, to give rise to an isogrammic
access sequence, can be simulated
in an oblivious fashion to have an I/O overhead that is
$O(\log n/\log B)$, with high probability.

There are a number of previous algorithm-engineering 
design paradigms that can facilitate
privacy-preserving data access in the cloud,
which, as we show, can be reduced to isogrammic access sequences
at only a constant cost per operation.
Thus, the observations made in this section may be of independent interest
for these specialized applications.

\subsection{Simulating Oblivious Data Structures}
The first application we explore is for bounded-degree directed data
structures in the \emph{oblivious data structure} framework of Wang
{\it et al.}~\cite{Wang:2014}. This framework applies to any data structure
that has a small number of ``root'' nodes for tree structures with 
bounded out-degree, such
that updates and accesses are done as a sequence of linked nodes starting
from a root.
Using
a heuristic similar to that used by Wang {\it et al.}~\cite{Wang:2014},
we can make any such access sequence isogrammic.
Namely, let us keep a random nonce, $r_u$, of $\lceil\log n\rceil$
bits for each node, $u$, in our data structure,
and let us assign the key for
accessing a node $u$ to be the pair, $(r_u,u)$.
That is,
any other node that points to $u$ will identify $u$ using the pair $(r_u,u)$.
The important observation is that any access sequence can inductively be able
to access nodes with their nonces, because bounded-degree data
structures are accessed starting from a root node; hence,
any set of node updates performs its operations on a path from a root and we can
update each node on such a path in reverse order 
to have its new random nonce. More
importantly, for each node, $x$, that points to node $u$,
we can also update $x$ to change its pointer to $u$ to have $u$'s
new nonce.
Using such nonces as keys, therefore, gives rise to an
isogrammic access sequence; hence, our result from
Theorem~\ref{thm:iso-fusion} applies to such
scenarios.
This implies the existence of efficient 
oblivious simulations for access sequences
involving stacks, queues, and deques (which have just $\Theta(1)$ node updates
per operation), as well as binary trees, such
as AVL trees and red-black trees (whose updates and searches can be padded
to have $\Theta(\log n)$ node updates per operation).
We summarize as follows.

\begin{theorem}
Any algorithm, $A$, written in the oblivious data structure framework,
with bounded-degree tree nodes reachable from a constant number of ``root'' nodes,
can be implemented as an algorithm, $A'$,
in the isogrammic algorithm design paradigm such that each data access 
in $A$ is translated into $O(1)$ accesses in $A'$.
\end{theorem}

For example, we can create an isogrammic queue, $Q$, by using an array, $A$,
of size $n$ and a dummy array, $B$, of size $n$, together with three global
``root'' indices, front, rear, and dummy.  Each time we access $Q$,
for an enqueue, dequeue, or no-op operation, we read 
the front, rear, and dummy variables.
If this is a no-op operation (or this would be an error operation, like doing 
a dequeue from an empty queue), then we next read the next dummy slot in $B$
and we increment the dummy variable.
If this is a valid enqueue, then we increment rear and add the new element
to that location in $A$.
If this is a valid dequeue, then we increment front and read the element from
the previous front location in $A$.
Anytime we wrap around $A$ or $B$, we increment a version counter (associated 
with the global variables), so that we access the cells in $A$ or $B$ by
index and version counter.
This implies that we always access the queue using an isogrammic access
sequence.

\subsection{Compressed Scanning}
In addition, the \emph{compressed-scanning} 
paradigm~\cite{Goodrich2013,Goodrich2014b}
also falls into our framework for isogrammic access sequences.
A compressed-scanning algorithm consists of $t$ rounds, where each round
involves accessing each of the elements of a set, $S$, of $n$
data items exactly once in a read-compute-write operation. 
By introducing random nonces and assigning them to items
for each round, we can easily
transform any algorithm in the compressed-scanning model into an isogrammic
access sequence.
Thus, all of the graph algorithms presented in these 
papers~\cite{Goodrich2013,Goodrich2014b}
can be simulated obliviously with our isogrammic OS scheme. 
We summarize as follows.

\begin{theorem}
Any algorithm, $A$, written in the compressed-scanning framework
can be implemented as an algorithm, $A'$,
in the isogrammic algorithm design paradigm such that each data access 
in $A$ is translated into $O(1)$ accesses in $A'$.
\end{theorem}

\section{Conclusion}
In this paper we have shown how to improve the I/O overhead for statistically 
secure ORAM simulation for a wide range of parameterized 
values for the client-side private memory size, $M$, and the size, $B$, of message
blocks.  Our results imply I/O overhead bounds that range from
$O(1)$ to $O(\log^2 n/(\log\log n)^2)$, with high
probability.
For example, we
can achieve an I/O overhead of $O(\log n)$, with high probability,
for statistically secure ORAM simulation 
if $B$ and $M$ are at least $\Omega(2^{\sqrt{\log n}})$, which is
asymptotically smaller than $n^\epsilon$, for any constant $0<\epsilon\le 1$.

For future work, it would be interesting to see if there is a super-logarithmic
lower bound for the I/O overhead for ORAM simulation for cases when $B$
and $M$ are small.
Also, it would interesting to see if it is possible to achieve an I/O
overhead for statistically secure ORAM simulation that is $O(\log n)$ when
$B$ and $M$ are $o(2^{\sqrt{\log n}})$, i.e., asymptotically smaller
than 
$2^{\sqrt{\log n}}$.

\section*{Acknowledgments}
This research was supported in part by
the National Science Foundation under grant 1228639.
This article also reports on work supported by the Defense Advanced
Research Projects Agency (DARPA) under agreement no.~AFRL FA8750-15-2-0092.
The views expressed are those of the authors and do not reflect the
official policy or position of the Department of Defense
or the U.S.~Government.
We would like to thank Eli Upfal and Marina Blanton
for helpful discussions regarding the topics
of this paper.

{\raggedright 
\ifFull
\bibliographystyle{abbrv}
\fi
\bibliography{../refs,../iso} 
}

\end{document}